\documentclass{article}
\usepackage{amsmath, amsthm, amsfonts,amssymb,latexsym}

\usepackage[T1]{fontenc}

\usepackage{enumitem}

\usepackage{biblatex}
\usepackage{hyperref}

\usepackage{appendix}

\newtheorem{lemma}{Lemma}
\newtheorem*{corollary}{Corollary}
\newtheorem{theorem}{Theorem}
\theoremstyle{definition}
\newtheorem*{remark}{Remark}
\newtheorem{example}{Example}

\theoremstyle{definition}
\newtheorem{definition}{Definition}[section]

\newcommand{\natnums}{\mathbb{N}}
\newcommand{\intnums}{\mathbb{Z}}
\newcommand{\ratnums}{\mathbb{Q}}
\newcommand{\reals}{\mathbb{R}}

\newcommand{\probP}{\mathbb{P}}
\newcommand{\expectE}{\mathbb{E}}

\newcommand{\adomain}{\mathbf{D}}

\newcommand{\Pois}{\operatorname{Pois}}
\newcommand{\Cauchy}{\operatorname{Cauchy}}
\newcommand{\kl}{\operatorname{kl}}
\newcommand{\KL}{\operatorname{KL}}

\newcommand{\Eta}{H}

\newcommand{\e}{\mathrm{e}}

\newcommand{\uut}{\mathbf{t}}

\usepackage{indentfirst}

\title{E-variables and tests of randomness for distribution classes}

\usepackage{authblk}
\author[1]{Georgii Potapov}
\author[1]{Yuri Kalnishkan}
\affil{Royal Holloway, University of London}

\date{February 2026}

\begin{document}

\maketitle

\begin{abstract}
    $E$-variables are a relatively new approach for testing statistical hypotheses that has been experiencing major development during the last several years.
    
    In this paper we introduce the method of $e$-variable-approximability and use it to develop a general approximation technique allowing us to construct $e$-variables for popular distribution classes important for applications.

    $E$-variables were originally based on a concept of Levin's (average-bounded) randomness tests from Algorithmic Information Theory. We show that our construction of $e$-variables can be used to provide an explicit construction for a randomness test with respect to a class of distributions.
\end{abstract}

\section{Introductions}

\subsection{History}

The problem of statistical hypothesis testing, that is, the problem of deciding whether given data can be described by a certain hypothesis, is naturally important in scientific research as well as applications.

Traditionally, hypothesis testing is done using $p$-values. Intuitively, a $p$-value is the probability under the null hypothesis that a chosen statistic on the given data is at least as extreme as its observed value. This defines a derivative statistic that we would call $p$-variable. Despite being widely used, $p$-values lack certain ``natural'' properties (for example, the value of a $p$-variable on a randomly chosen initial segment of the data, in general, is not a $p$-variable itself), making them easy to accidentally misuse.

\emph{$E$-variables} are statistical tests that can be viewed as a flexible and versatile alternative to $p$-variables for statistical hypothesis testing. Formally, an $e$-variable is a non-negative statistical test with expected value at most $1$ (\cite{vovk2020nonalgorithmic}).

Their use allows for a flexible experiment design that allows to test against the null for given alternative hypothesis and lacks many disadvantages of the commonly used $p$-values. $E$-variables are easy to combine, which simplifies meta-analysis, and running product of $e$-variables on independent bathces of data is a supermartingale, and thus continues to be an $e$-variable at random stopping time (\cite{VovkWang2021Calibration, GrunwaldSafeTesting}).

In this work we revisit the conceptual origin of $e$-variables --- \emph{randomness deficiency}. Randomness deficiency, introduced by Levin (\cite{levin1976uniform}, \cite{LEVIN198415}) is a concept from Algorithmic Information Theory, which can be loosely defined as an optimal test of randomness which exponent is an $e$-variable. Informally speaking, randomness deficiency is a function that counts number of regularities of its argument, a random object, and serves as a measure of how implausible an outcome of an experiment is (\cite{shen}).

The papers \cite{Kolmogorov1968LogicalBF} and \cite{vovk2016concept} give different but related definitions of a random \textit{``Bernoulli bit sequence'' }in terms of algorithmic randomness. The differences in these definitions correspond to differences between distributions of sequences with i.i.d. bits and exchangeable distributions, and can be quantified using randomness deficiency of number of $1$s in the sequence. 

One of the results of papers \cite{vovk2016concept} and \cite{VovkLearningBernoulli199796}, mentioned in our paper as Theorem \ref{thm:vovk-bern}, states that, in terms of randomness deficiency, it is approximately as plausible for an integer $k$ between $0$ and $n$ to be obtained from \emph{some} binomial distribution as it would be plausible for $k$ to be obtained from the uniform distribution on integer intervals $[s_j:s_{j+1}-1]\ni k$, where $\{s_1 < s_2 \ldots < s_m\}$ is some set that depends only on $n$.

Randomness tests relative to a class of distributions were more generally studied in \cite{Bienvenu2011AlgorithmicTA}, where their existence was proven from a general topological argument for the Cantor space and effectively closed classes of measures.

\subsection{Our contribution}

One of the goals of our research is to generalise the results of \cite{vovk2016concept} and \cite{Bienvenu2011AlgorithmicTA} to make them more applicable to practical machine learning.

Regarding algorithmic randomness, our main results are theorems \ref{thm:rand-def-of-a-family} and \ref{thm:rand-def-norm-cauchy}, which claim lower semi-computability of randomness tests for a class of distribution for several important distribution families, such as uniform distributions on intervals, Poisson distributions, normal (Gaussian) distributions and Cauchy distributions. Furthermore, our results explicitly express randomness with respect to a class in terms of universal uniform tests of randomness.

Compared to the similar statement in \cite{Bienvenu2011AlgorithmicTA}, we work in $\reals^n$ rather than on the Cantor set. Our results are proven for classes of distributions that do not constitute topologically closed sets in the space of probability distributions (for example, the sequence of uniform distributions on $[0:n]$ should converge to the measure which is zero on all numbers, which is not a uniform measure on a finite integer segment).

Theorems \ref{thm:rand-def-of-a-family} and \ref{thm:rand-def-norm-cauchy} follow from theorems \ref{thm:most-e-var-results} and \ref{thm:contin-e-var}, which can be seen as their non-algorithmic, purely statistical counterparts. We provide a uniform way to turn any set of $e$-variables for simple hypotheses $\{\probP_\theta\}$ into an $e$-variable for the composite hypothesis $\{\probP_\theta\}_{\theta\in\Theta}$. Our construction is not based on either reverse information projection or calibration of $p$-variables, and can use any $e$-variables for simple hypotheses as black boxes to produce a valid $e$-variable for the corresponding complex hypothesis.

Both our statistical and algorithmical results are built on the use of complexity reducing parameter estimation, which can be seen as another theoretical justification for minimum description length principle. 






\section{Preliminaries}

\subsection{E-variables}

$P$-values (and the corresponding statistics that we call $p$-variables) are a popular instrument for statistical hypothesis testing.

In the following definitions, $\mathcal{H}$ is a set of probability distributions (from \textbf{h}ypothesis).

\begin{definition}[$p$-variable]
    We say that a statistic (measurable function) $P$ is a \emph{p-variable}, if for every distribution $\probP\in\mathcal{H}$ we have
    \[
        \mathbb{P}(P(X) \leq p) \leq p
    \]
    for all $p \in [0,1]$.
\end{definition}

\begin{definition}[$e$-variable]
    We say that a statistic $E: \mathcal{X} \to [0,+\infty]$ is an \emph{$e$-variable} under hypothesis $\mathcal{H}$, if
    \[
        \expectE{E}_{\probP}(X) \leq 1.
    \]
    for every probability distribution $\probP\in\mathcal{H}$.
\end{definition}

While $p$-variables are more widely used in science, if $P$ is a $p$-variable and $E$ is an $e$-variable, then $1/E$ is a $p$-variable and $\kappa P^{\kappa-1}$ is an $e$-variable for all $\kappa\in(0,1)$. Conversion between $p$- and $e$-variables is discussed in details in \cite{VovkWang2021Calibration}.

In the case of a simple null hypothesis and a simple alternative hypothesis with probability densities $p_0$ and $p_1$, the ratio $p_1(x)/p_0(x)$ would be an $e$-variable for the null. This can be generalised for the case of complex hypotheses using reverse information projection (\cite{LardyReverseInfProj, GrunwaldSafeTesting}).

Theory of $e$-variables was developed as an adaptation of theory of uniform tests of randomness, defined in the following subsection, for the general case of measurable functions and eventually has become an its own field of research (\cite{vovk2020nonalgorithmic, RamdasWang2025}).

\subsection{Computable and semicomputable functions}

We use a standard definition of computable and lower (upper) semicomputable functions that follows \cite{GACS200591} and \cite{Bienvenu2011AlgorithmicTA}. This approach is in agreement with other works on the topic (such as \cite{VovkVyuginBayes}), and generally is associated with the Type-2 theory of computability as described in \cite{weihrauch2000computable}.

\begin{definition}\label{dfn:lsc-fun}
    A function $f : \adomain \to [-\infty,\infty]$ is \emph{lower semicomputable} (l.s.c.) if its undergraph $\{(x,t) : t < f(x) \}$ can be presented as a union $\bigcup_{(m,n)\in W} B_m \times [-\infty,r_n)$ with $W \subseteq \natnums\times\natnums$ an enumerable set and $(r_i)_i$ and $(B_i)_i$ being pre-specified enumerations of $\ratnums$ and of a countable topological base of $\adomain$, respectfully.

    We also say that $f$ is upper semicomputable (u.s.c.) if $-f$ is l.s.c., and $f$ is computable when it is both l.s.c. and u.s.c.
\end{definition}

A more general definition of computability (needed for theorems \ref{thm:rand-def-of-a-family} and \ref{thm:rand-def-norm-cauchy}) is as follows.

\begin{definition}
    In general, we say that a function $f: X \to Y$ is computable if, for some prespecified enumerations of countable bases $\mathbf{B}^X=\left(B_1^X,B_2^X,\ldots\right)$ and $\mathbf{B}^Y=\left(B_1^Y,B_2^Y,\ldots\right)$ of $X$ and $Y$ respectfully we have
\[
    \bigcup_{k} f^{-1}\left(B_k^Y\right) \times B_k^Y = \bigcup_{(m,n)\in W} B_m^X \times B_n^Y
\]
for some enumerable set $W$.
\end{definition}

\begin{remark}
    Note that lower semicomputability implies lower semicontinuity, upper semicomputability implies upper semicontinuity and computability implies continuity.
\end{remark}

In this paper, we consider (semi)computability of functions defined either on $\reals \times \mathfrak{P}(\reals)\times\{0,1\}^*$ or $\natnums \times \mathfrak{P}(\natnums)\times\{0,1\}^*$. Slightly informally, we can think that function $f$ is lower semicomputable on such domain if there is an oracle Turing machine that approximates $f(u)$ from below and accesses ``infinite'' components of its input $u=(x,\probP,s)$ through oracles. 

Oracle for a real number or a distribution, on a rational input $\varepsilon >0$, returns an $\varepsilon$-appoximation of the object.

An $\varepsilon$-approximation of a probability measure $\probP$ would be a measure $\probP_\varepsilon=\sum_{i=1}^n p_i \delta_{x_i}$ concentrated on a finite number of rational points $x_i$ and assigning them rational probabilities $p_i$, such that $\probP_{\varepsilon}$ is $\varepsilon$-close to $\probP$ in L\'evy–Prokhorov metric, that is, for all measurable sets $A$ we have
\[
    \probP(A) < \probP_\varepsilon(A^\varepsilon) + \varepsilon,~~~ \probP(A^\varepsilon) > \probP_\varepsilon(A) - \varepsilon,
\]
where $A^\varepsilon$ is the $\varepsilon$-neighbourhood of $A$.




\subsection{Randomness tests}

\begin{definition}[Randomness test]\label{dfn:rand-test}

For a measurable space $\mathcal{X}$, consider a function $f(\cdot | \cdot, \cdot):\mathcal{X}\times\mathfrak{P}(\mathcal{X})\times\mathcal{Y}\to[0,\infty]$, which takes as its arguments an object $x \in \mathcal{X}$, a probability distribution $\probP$ on $\mathcal{X}$ and an auxiliary object in  $y \in \mathcal{Y}$. We say that $f$ is a \emph{uniform randomness test} (or a test of randomness in a sense of Levin), if it is:
\begin{enumerate}
    \item a lower semicomputable function (assuming an appropriate computable structure on $\mathcal{X}, \mathfrak{P}(\mathcal{X})$ and $\mathcal{Y}$),
    \item for all probability measures $\probP$ and all strings $y$ the function $f(\cdot\mid \probP, y)$ is an $e$-variable for the simple hypothesis $\{\probP\}$.
\end{enumerate}
\end{definition}

From a practical point of view, lower semi-computability is a very mild restriction that allows us to construct a universal test of randomness (described below). Philosophically, lower semi-computability represents the idea that, if an $e$-variable rejects the null with a chosen significance level, than this must be eventually established from the data measured with high enough but finite precision.

The next theorem-definition is a well known result.

\begin{theorem}[existence and definition of a universal uniform test (\cite{levin1976uniform, GACS200591, VovkVyuginBayes})]
For $\mathcal{X}$ that is of the form $\{0,1\}^*, \{0,1\}^\omega, \reals^n$ and $\mathcal{Y}$ that is either one of these sets or the set of bounded measures on such sets, there exists a uniform randomness test $\uut$ such that for any uniform randomness test $f$ there is a constant $C_f >0$, depending only on $\uut$ and $f$, for which the inequality
\[
    \uut(x\mid \probP, y) \ge C_f f(x\mid \probP, y)
\]
holds for any $x, \probP, y$.

Any such test is called a \emph{universal uniform randomness test}, or simply a \emph{universal test}.
\end{theorem}

From now on we are going to use symbols $=^+, \le^+$, $\ge^+$, $=^\times$, $\le^\times$, $\ge^\times$ to denote that an (in)equality holds up to an additive or a multiplicative constant.

For any universal tests $\uut_1$ and $\uut_2$ we have $\uut_1=^\times \uut_2$, so for the rest of the paper we fix an arbitrary universal test and refer to it as $\uut$.

The current paper is mainly concerned with studying the following object.

\begin{definition}
    For a family of probability distributions $\mathcal{P}$, we define \emph{randomness test relative to $\mathcal{P}$} as the \emph{$\inf$-projection}
    \[
        \uut^{\mathcal{P}}(x\mid y) = \inf_{\probP\in\mathcal{P}}\uut(x\mid \probP, y).
    \]
\end{definition}


In this paper we are concerned with finding explicit expressions for $\inf$-projections, thus directly proving that randomness test for a class is a l.s.c. function.

We will show this for many important families of distributions $\mathcal{P}$, listed below.

Our technique generalises and extends the approach of \cite{VovkLearningBernoulli199796, vovk2016concept}, where the following theorem is proved.

\begin{theorem}[\cite{VovkLearningBernoulli199796, vovk2016concept}]\label{thm:vovk-bern}
    For any natural number $n$ let $\operatorname{Bern}_n=\{\operatorname{Bern}(n,p)\}_{p\in[0,1]}$ be the family of distributions on $\{0,1\}^n$ with i.i.d. bits, and $\operatorname{Bin}_n = \{\operatorname{Bin}(n,p)\}_{p\in [0,1]}$ be the family of binomial distributions on $[0:n]=\{0,1,\ldots,n\}$. Then for any string $x \in \{0,1\}^n$ and any number $k \in [0:n]$
    \begin{align*}
        \uut^{\operatorname{Bern}_n}(x) &=^\times \uut(x\mid \operatorname{Bern}(n,\hat p(|x|_1/n))),\\
        \uut^{\operatorname{Bin}_n}(k) &=^\times 
        \uut(k\mid \operatorname{Bern}(n,\hat p(k/n)))\\
        &=^\times
        \uut\left(k \mid \ \mathcal{U}\left(\hat p^{-1} (\hat p(k/n)\right)\right),
    \end{align*}
    where $\hat p$ is a computable function that maps a rational number of the form $m/n$, ${m \in [0:1]}$, to the nearest element of the ``net'' ${\left\{ \sin^2 \left( \frac{\pi t}{2\lfloor \sqrt{n}\rfloor}\right) \right\}_{t\in [1:\lfloor\sqrt{n}\rfloor]}}$.
\end{theorem}

We are going to prove similar results for the following families of distributions: 
\begin{itemize}
    \item Uniform distribution on $[0:n]=\{0, 1, \ldots, n\}$, $\{\mathcal{U}([0:n])\}_{n\in\natnums}$,
    \item Poisson distributions, $\{\Pois(\lambda)\}_{\lambda>0}$,
    \item uniform distributions on real intervals, $\{\mathcal{U}([0,\theta])\}_{\theta > 0}$,
    \item normal distributions with fixed variance $\{\mathcal{N}(\mu,1)\}_{\mu\in\reals}$,
    \item normal distributions with fixed mean value $\{\mathcal{N}(0,\sigma^2)\}_{\sigma^2>0}$,
    \item Cauchy distributions with fixed scale $\{\Cauchy(x_0,1)\}_{x_0\in\reals}$.
\end{itemize}

Our methods differ significantly from those in \cite{VovkLearningBernoulli199796, vovk2016concept}. We are going to prove our results (formulated precisely in section \ref{sec:rand-def-res}) by showing that all the distribution families in question are what we call $e$-variable-approximable, which by itself mean that they admit a certain construction of $e$-variable for the family from $e$-variables for individual distributions.  

\subsection{Organisation}

In section \ref{sec:e-var-approx}, we are going to formulate a number of technical results on $e$-variables that do not include any computability requirements: subsection \ref{sec:e-var-approx} states the definition of $e$-variable-approximability and lists the main results about this notion, subssection \ref{sec:interp-appr} describes continuously interpolated versions of $e$-variables constructed in subsection \ref{sec:e-var-approx}, which we would need for making our functions l.s.c. in section \ref{sec:rand-def-res}. 

In section \ref{sec:rand-def-res}, we will use these results for $e$-variables to construct l.s.c. approximations of $\inf$-projections for the claimed families of distributions. Most of the proofs are deferred to the appendix. 

In section \ref{sec:discs} we discuss our choice of definitions, scale of applicability of our results and perspectives for further research.

All non-standard notation is listed in the appendix.

\section{$E$-variable-approximability}\label{sec:e-var-approx}

\subsection{Main results on $e$-variable-approximability}

Here we introduce one of the main concepts in our work, the notion of $e$-variable approximability of a distribution family. 

\begin{definition}\label{dfn:e-val-approx}
    We say that a family $\mathcal{H}$ of distributions is \emph{$e$-variable approximable (from the net $\mathcal{S}$) (with factor $C$)}, if there exists an at most countable set $\mathcal{S}$, called \emph{a net}, a set of distributions $\{\probP_s\}_{s\in \mathcal{S}} \subset \mathcal{H}$, a function $\hat s : \reals \to \mathcal{S}$ that we would call \emph{an estimator} and a constant $C > 0$ such that for any family of non-negative statistical tests $\{e_s\}_{s\in\mathcal{S}}$, if for all $s$ the test $e_s$ is an $e$-variable under the simple hypothesis $\{ \probP_{s}\}$, then the test $e(x) = \frac{1}{C} e_{\hat s(x)}(x)$ is an $e$-variable under the hypothesis $\mathcal{H}$.
\end{definition}

This definition might seem somewhat artificial at first, and so we provide some argumentation and examples in Section \ref{sec:discs}.

The following lemmas provide easily verifiable conditions for establishing $e$-variable approximability. 

\begin{lemma}\label{lem:4-cond-techn-lemma}
    Let $\mathcal{H} = \{\probP_\theta\}_{\theta\in\Theta}$ be a family of distributions on $\reals^n$ with densities $\{p_\theta\}_{\theta\in\Theta}$, parametrised by the elements of $\Theta\subseteq\reals$, and let $\mathcal{S} \subseteq \Theta$ be an at most countable set such that $\operatorname{succ}_{\mathcal{S}}(\theta)=\min\{s\in\mathcal{S}\mid\theta<s\}$ and $\operatorname{pred}_{\mathcal{S}}(\theta)=\max\{s\in\mathcal{S}\mid s<\theta\}$ are well-defined for all $\theta\in\Theta$ expect for possibly those that are greater than maximal or less than minimal element of $\mathcal{S}$. Let $d(\cdot\|\cdot):\Theta\times\Theta\to [0, +\infty]$ be a well-defined function. Let $\hat s:\reals^n\to \mathcal{S}$ be a statistic and $g:\reals^n\to\Theta$ be a function.
    Consider now the following properties:
    \begin{enumerate}[label=(p \roman*)]
        \item\label{prop:loglike-to-d} For every $\theta \in\Theta$, $s \in \mathcal{S}$ and $x\in \hat{s}^{-1}(s) \cap \{y\mid p_\theta(y)>0\}$ we have that
        \[
            \log\frac{p_\theta(x)}{p_{s}(x)} = d(g(x)\|s)-d(g(x)\|\theta);
        \]
        \item\label{prop:near-bound} there is a constant $c' > 0$ such that for any $s \in \mathcal{S}$ and any $x$ for which $x \in \hat s^{-1}(s)$ we have
        \[
            d(g(x)\|s) < c'.
        \]
        \item\label{prop:g-near-s} For all $s\in\mathcal{S}$ and all $x \in \hat s^{-1}(s)$
        \[
            \operatorname{pred}_{\mathcal{S}}(s)\le \operatorname{pred}_{\mathcal{S}}(g(x))\le \operatorname{succ}_{\mathcal{S}}(g(x))\le \operatorname{succ}_{\mathcal{S}}(s).
        \]
        \item\label{prop:d-grows-fast-enough} There is a constant $\alpha >0$ such that for all $\theta_1, \theta_2 \in \Theta, \theta_1 < \theta_)$, it holds that
        \[
            d(\theta_1\|\theta_2) \ge (1+\alpha)\log (k-1), ~~~ d(\theta_2\|\theta_1) \ge (1+\alpha)\log (k-1),
        \]
        for $k = \left| \mathcal{S} \cap (\theta_1, \theta_2) \right|$\footnote{for this statement, we treat the logarithm of non-positive numbers as $-\infty$}.
    \end{enumerate}

    If all the properties \ref{prop:loglike-to-d}--\ref{prop:d-grows-fast-enough} of the data $(\mathcal{H}, \mathcal{S}, \Theta, \hat{s}, g, d)$ hold, then $\mathcal{H}$ is $e$-variable approximable from the net $\mathcal{S}$ and the estimator $\hat s$, and the constant factor $C$ can be taken as $C =\exp(c')\left(7 +  2/\alpha\right)$.
\end{lemma}

\begin{lemma}\label{lem:5-cond-techn-lemma}
    In lemma \ref{lem:4-cond-techn-lemma}, condition \ref{prop:d-grows-fast-enough} follows from the following two conditions:
    \begin{enumerate}[label=(p \roman*')]
        \item\label{prop:obtuse-pyth} if $\theta_1 \leq \theta_2 \leq \theta_3$ or $\theta_1 \ge \theta_2 \ge \theta_3$, then 
        \[
           d(\theta_1\|\theta_3) \ge d(\theta_1\|\theta_2) + d(\theta_2\|\theta_3);
        \]
        \item\label{prop:fixed-steps} there is a constant $c > 0$ such that for any consecutive $s, s' \in \mathcal{S}$, it holds that
        \[
            d(s\|s') > c,\quad d(s'\|s) > c.
        \]
    \end{enumerate}

    If all the properties \ref{prop:loglike-to-d}--\ref{prop:g-near-s} of the data $(\mathcal{H}, \mathcal{S}, \Theta, \hat{s}, g, d)$ in lemma \ref{lem:4-cond-techn-lemma} hold, together with \ref{prop:obtuse-pyth} and \ref{prop:fixed-steps}, then $\mathcal{H}$ is $e$-variable approximable from the net $\mathcal{S}$ and the estimator $\hat s$, and the constant factor $C$ can be chosen as $C = \exp(c')\left(5 + \frac{2}{\e^c - 1}\right)$.
\end{lemma}

Both lemmas are proven by carefully estimating the terms in the right-hand side of
\[
    \expectE_{\theta} e_{\hat s(X)}(X) = \sum_{s\in\mathcal{S}} \int_{\hat s^{-1}(s)} e_s(x) \frac{p_\theta (x)}{p_s(x)} p_s(x) {\rm d}x
\]
to provide a uniform upper bound on the sum.

We will see that in many cases checking the conditions of either of two lemmas is straightforward.

In particular, it can be easily checked for many standard exponential families.

\begin{lemma}\label{lem:exp-fams-are-good}
    Let $\{ \probP_\eta\}_{\eta\in\Eta}$, with $\Eta$ an open (possibly infinite) interval in $\reals$, be an exponential family canonically parametrised by $\eta$ with support $D \subset\reals^n$ and densities $p_\eta(x) = h(x)\exp\left(\eta T(x) - A(\eta)\right)$.

    \begin{enumerate}
        \item If there is a function $\hat\eta:D\to \Eta$ that solves the likelihood equation $T(x) = \expectE_{\hat\eta(x)}T(X)$, for all $x \in D$, then condition of \ref{prop:loglike-to-d} of lemma \ref{lem:4-cond-techn-lemma} holds for $d = \kl$ and $g=\hat\eta$.
        \item The condition \ref{prop:obtuse-pyth} of lemma \ref{lem:5-cond-techn-lemma} holds for $d = \kl$.
    \end{enumerate}
\end{lemma}

\begin{theorem}\label{thm:most-e-var-results}
    The following families of distributions are $e$-variable-approximable. The function $\hat{s}$ is omited whenever $\hat{s}=\operatorname{round}_{\mathcal{S}}$.
    \begin{itemize}
        \item $\{\operatorname{Bin}(n,p)\}_{p\in(0,1)}$, from the net $\mathcal{S}=\left\{ \sin^2\left(\frac{\pi t}{2\lfloor \sqrt{n}\rfloor}\right)\right\}_{0<t<\sqrt{n}}$;

        \item $\{\mathcal{U}(\{0,1,\ldots,n\})\}_{n\in\mathbb{N}}$, from the net $\mathcal{S}=\{2^n\}_{n\in\mathbb{N}}$, $\hat s(n)=2^{\lceil \log_2 n \rceil}$;

        \item $\{\operatorname{Poiss}(\lambda)\}_{\lambda\in\mathbb{R}_{>0}}$, from the net $\{t^2\}_{t\in\mathbb{N}}$.
    \end{itemize}

    The following families of continuous distributions are also $e$-variable-approximable:
    \begin{itemize}
        \item $\{\mathcal{U}([0,\theta])\}_{\theta > 0}$, from the net $\mathcal{S}=\{2^n\}_{n\in\intnums}$, $\hat s(x)=2^{\lceil \log_2 x \rceil}$
        \item $\{\mathcal{N}(\mu, 1)^{\otimes n}\}_{\mu\in\mathbb{R}}$, from the net $\mathcal{S}=\frac{\alpha}{\sqrt{n}}\mathbb{Z}$ for $\alpha > 0$, $\hat{s}(x_{1:n})=\operatorname{round}_\mathcal{S}\left(\frac{1}{n}\sum_k x_k\right)$, the factor may depend on $\alpha$ but does not depend on $n$;
        \item $\{\mathcal{N}(0, \sigma^2)^{\otimes n}\}_{\sigma\in\mathbb{R}_{>0}}$, from the net $\mathcal{S}=\left\{\left(1 + \frac{1}{\sqrt{n}}\right)^k\right\}_{k\in\mathbb{Z}}$, $\hat{s}(x)=\operatorname{round}_\mathcal{S}\left(\frac{\|x\|^2}{n}\right)$, the factor does not depend on $n$;
        \item $\{ \operatorname{Cauchy}(x_0,1)\}_{x_0 \in \reals}$, from the net $\mathcal{S}=\intnums$.
    \end{itemize}

    For any $\varepsilon \le 1/5$ and let $r^\varepsilon$ be any function defined on $\reals\setminus\left(\intnums+0.5\right)^\varepsilon$ as $\operatorname{round}_\intnums$, and on for each $\varepsilon$-neighbourhood of a half-integer $n+0.5$ defined as either $n$ or $n+1$ (maybe differently for each $n$), then $\{\mathcal{N}(\mu,1)\}_{\mu\in\reals}$ and $\{\Cauchy(x_0,1)\}_{x_0\in\reals}$ are both $e$-variable approximable from $\mathcal{S}=\intnums$ and $\hat{s}=r^\varepsilon$. 
\end{theorem}

\begin{corollary}
    For all $n$ and $x_{1:n} \in \reals^n$, let $r_n(x_{1:n})=\operatorname{round}_{\frac{1}{\sqrt{n}}\intnums}\left(\frac{1}{n}\sum_k x_k\right)$. For some constant $C > 0$, the following holds.
    
    For any $\{e_\mu\}_{\mu\in\reals}$, a family of measurable functions, such that each $e_\mu$ is an $e$-variable w.r.t. $\{\mathcal{N}(\mu,1)\}$, the function
    \[
        e(x_{1:n}) = \frac{1}{C}\prod_{k=1}^n e_{r_n(x_{1:n})}(x_k)
    \]
    is an $e$-variable for $\{\mathcal{N}(\mu,1)^{\otimes n}\}_{\mu\in\reals}$.
\end{corollary}

The corollary immediately follows from Theorem \ref{thm:most-e-var-results} and the fact that a product of $e$-variables is an $e$-variable for the product distribution.

\subsection{Interpolating approximation}\label{sec:interp-appr}

Discontinuity of $\operatorname{round}_{\mathcal{S}}$ makes this function uncomputable for distributions on real numbers, which prevents us from applying our results directly to randomness deficiency.

For $\operatorname{round}_{\reals}$ we would like to ``smoothen'' the transition  between $n$ and $n+1$, and similarly for $\operatorname{round}_\mathcal{S}$ in general.

\begin{definition}
    We say that the family $\{\probP_\theta\}_{\theta\in\Theta}$ is \emph{continuously $e$-variable-approximable} from the net $\mathcal{S} \subset \Theta$ if there is a set of functions $\{\hat s(\cdot,s)\}_{s\in\mathcal{S}}$, constituting a partition of unity, and a constant $C > 0$ such that a function $e$ defined as
    \[
        e(x) = \frac{1}{C}\sum_{s\in \mathcal{S}} e_s(x) \hat s(x, s)
    \]
    is an $e$-variable for $\mathcal{P}$ whenever all $e_s$ are $e$-variables for corresponding $\{\probP_s\}$.
\end{definition}

\begin{theorem}\label{thm:contin-e-var}
    For $\mathcal{P} = \{\probP_\theta\}=\{\mathcal{N}(\theta,1)\}_{\theta\in\reals}$ or $\{\Cauchy(\theta,1)\}_{\theta\in\reals}$, for any $\{e_n\}_{n\in\intnums}$ such that $e_n$ is an $e$-variable for $\{\probP_n\}$, there is a constant $C>0$ such that for any $\varepsilon \le 1/5$ the function
    \[
        e^\varepsilon(x) = \frac{1}{C} \sum_{n=-\infty}^\infty e_n(x) \mathbb{I}_{\left[n-\frac{1}{2}+\varepsilon, n + \frac{1}{2}-\varepsilon\right]}^{2\varepsilon}(x)
    \]
    is an $e$-variable for $\mathcal{P}$.
\end{theorem}

Note that for any $x$ at most two terms in the sum are non-zero.

The proof illustrates flexibility of $e$-variable approximability property.

\begin{proof}
	Notice that an $e$-variable multiplied by a continuous $[0,1]$-valued function is also an $e$-variable
    For a given family $\{e_n\}_{n\in\intnums}$, define two new families: $\{e_n^{\textrm{even}}\}_{n\in\intnums}$ and $\{e_n^{\textrm{odd}}\}_{n\in\intnums}$ in the following way. Let $e_n^{\textrm{even}}$ be equal to $e_n\cdot \mathbb{I}^{2\varepsilon}_{[n-0.5-\varepsilon, n+0.5 +\varepsilon]}$ for even $n$, and equal to $0$ for odd $n$; similarly for $e_n^\textrm{odd}$.

    Let $\hat{s}_\textrm{even}$ be the rounding function $\operatorname{round}_{\intnums}$, redefined on $\varepsilon$-neighbourhoods of half-integers to round to the nearest even number. We define $\hat{s}_\textrm{odd}$ in a similar way.

    Then there is some $C_0 > 0$ such that both $e^\textrm{even}(x)=\frac{1}{C_0}e^\textrm{even}_{\hat{s}_\textrm{even}(x)}(x)$ and $e^\textrm{odd}(x)=\frac{1}{C_0}e^\textrm{odd}_{\hat{s}_\textrm{even}(x)}(x)$ are $e$-variables with respect to $\mathcal{P}$. Arithmetic mean of two $e$-variables is trivially an $e$-variable, and so for $C = 2C_0$ we obtain that
    \[
        e^\varepsilon(x)=\frac{e^\textrm{even}(x) + e^\textrm{odd}(x)}{2} = \frac{1}{C} \sum_{n=-\infty}^\infty e_n(x) \mathbb{I}_{\left[n-\frac{1}{2}+\varepsilon, n + \frac{1}{2}-\varepsilon\right]}^{2\varepsilon}(x)
    \]
    is also an $e$-variable w.r.t. $\mathcal{P}$.
\end{proof}

\section{Lower semicomputable approximations of randomness deficiency $\inf$-projection}\label{sec:rand-def-res}

We say that a family of distributions $\{\probP_\theta\}_{\theta\in\Theta}$ is computable if the function $(\theta \mapsto \probP_\theta)$ is computable.

\begin{lemma}\label{lem:semicomp-test-discrete}
    For a computable family of distributions $\mathcal{P}=\{\probP_\theta\}_{\theta\in\Theta}$ on a discrete set $\mathcal{X}$ that is $e$-variable-approximable from a net $\mathcal{S}$ with a computable estimator $\hat{s}$ and factor $C$, then there is a randomness test $f: \mathcal{X}\times\mathfrak{P}(\mathcal{X}) \times \mathcal{Y} \to [0,\infty]$ such that 
    \[
        f(x\mid\probP_\theta, y) =^\times \uut(x\mid \probP_{\hat{s}(x)}, y)
    \]
    for all $\theta\in\Theta$.
\end{lemma}

\begin{proof}[Sketch of the proof]
    The proof technique is standard for the field.
    Lower semi-computability of $\uut$ together with computability of $\probP$ allows us to enumerate a list of bounded continuous piece-wise linear functions $0=g_0 \le g_1\le\ldots$ with finite description, with $g_n(x, y)\uparrow \uut(x\mid\probP_{\hat s(x)}, y)$. For these function we can estimate integrals $\int g_n(x, y)\probP({\rm d}x)$. Without loss of generality, we say that $C \ge 1$ and rational. We can then define
    \[
        f(x\mid \probP, y) = \frac{1}{2} \sup_n \left\{g_n(x, y) \mid \int g_n(x, y)\probP({\rm d}x) < 2C\right\}.
    \]
    From $e$-variable-approximability it follows that the condition on integral being less than $2C$ is always true for $\probP \in \mathcal{P}$, and thus $f =^\times \uut$ in this case.
\end{proof}

\begin{theorem}\label{thm:rand-def-of-a-family}
    For a computable family $\mathcal{P} = \{\probP_\theta\}_{\theta\in\Theta}$ that is $e$-variable-approximable with a computable estimator $\hat{s}$, it holds that
    \[
        \uut^{\mathcal{P}}(x\mid y) =^\times \uut(x \mid \probP_{\hat s(x)}, y).
    \]
\end{theorem}

\begin{proof}
    Let $f$ be the uniform test from lemma \ref{lem:semicomp-test-discrete}.
    
    For any $\theta$
        \[
            f(x\mid \probP_\theta, y) =^\times \uut(x \mid \probP_{\hat s(x)}, y),
        \]
    where the right hand term does not depend on $\theta$.
    
    We get then, that for any $\theta$
    \[
        f(x\mid \probP_\theta, y) =^\times \uut(x\mid \probP_{\hat s(x)}, y) \ge \uut^{\mathcal{P}}(x\mid y).
    \]

    From optimality of randomness deficiency we also obtain that
    \[
        \uut(x\mid \probP_\theta, y) \ge^\times f(x \mid \probP_\theta, y) =^\times \uut(x \mid\probP_{\hat s(x)}, y).
    \]

    Since constants under $\ge^\times, =^\times$ do not depend on $x, \theta, y$, we can take the infimum over all $\theta$ and get
    \[
        \uut^{\mathcal{P}}(x\mid y) \ge^\times \uut(x \mid \probP_{\hat s(x)}, y),
    \]
    which concludes the proof.
\end{proof}

\begin{lemma}\label{lem:semicomp-test-contin}
    For a computable family of distributions $\mathcal{P}=\{\probP_\theta\}_{\theta\in\Theta}$ that is continuously $e$-variable-approximable from a uniformely computable net $\mathcal{S}=\{s_i\}_{i\in\mathcal{I}}$, $\mathcal{I}\subset \natnums$ with $\hat{s}$ being l.s.c. uniformely in both arguments, and factor $C$, then there is a randomness test $f: \mathcal{X}\times\mathfrak{P}(\mathcal{X}) \times \mathcal{Y} \to [0,\infty]$ such that 
    \[
        f(x\mid\probP_\theta, y) =^\times \sum_{s\in \mathcal{S}}\uut(x\mid \probP_s, y) \hat s(x, s)
    \]
    for all $\theta\in\Theta$.
\end{lemma}

\begin{lemma}\label{lem:distr-shift}
    Let $\mathcal{P}=\{\probP_\theta\}_{\theta\in\reals}$ be either $\{\mathcal{N}(\theta,1)\}_{\theta\in\reals}$  or $\{ \Cauchy(\theta,1)\}_{\theta\in\reals}$. Let $n$ be an integer and $x \in [n, n+1]$. Then
    \[
        \uut(x \mid \probP_n, y) =^\times \uut(x \mid \probP_{n+1}, y).
    \]
\end{lemma}

This lemma can be proven using the same technique as lemma \ref{lem:semicomp-test-discrete} after noticing that there is a computable function on distribution that on the elements of $\mathcal{P}$ recovers the parameter. An example of such function would be a function that computes $\expectE_\probP \arctan(X)$ for given $\probP$ and recovers the would-be parameter by assuming that $\probP$ is from $\mathcal{P}$.

\begin{theorem}\label{thm:rand-def-norm-cauchy}
    Let $\mathcal{P}=\{\probP_\theta\}_{\theta\in\reals}$ be either $\{\mathcal{N}(\theta,1)\}_{\theta\in\reals}$ or $\{\Cauchy(\theta,1)\}_{\theta\in\reals}$. Then
    \[
        \uut^{\mathcal{P}}(x\mid y) =^\times \uut(x \mid \probP_{\lfloor x\rfloor}, y).
    \]
\end{theorem}

\begin{proof}
    Let $f$ be the uniform test from lemma \ref{lem:semicomp-test-contin} and $\hat s(x, n) = \mathbb{I}^{0.2}_{[n-0.4, n+0.4]}(x)$ in accordance with theorem \ref{thm:contin-e-var}. Then, just as in the proof of theorem \ref{thm:rand-def-of-a-family} we can claim that
    \[
        \uut^{\mathcal{P}}(x\mid y) =^\times \uut(x \mid \probP_{\lfloor x \rfloor}, y) \hat s(x, \lfloor x \rfloor) + \uut(x \mid \probP_{\lfloor x \rfloor +1}, y) \hat s(x, \lfloor x \rfloor +1).
    \]

    Now, from lemma \ref{lem:distr-shift} we obtain that
    \[
        \uut(x \mid \probP_{\lfloor x \rfloor}, y) \hat s(x, \lfloor x \rfloor) + \uut(x \mid \probP_{\lfloor x \rfloor +1}, y) \hat s(x, \lfloor x \rfloor +1) \ge^\times \uut(x \mid \probP_{\lfloor x \rfloor}, y),
    \]
    and so
    \[
        \uut^{\mathcal{P}}(x \mid y) \ge^\times \uut(x \mid \probP_{\lfloor x \rfloor}, y) \ge \uut(x \mid y).
    \]
\end{proof}





\section{Discussions}\label{sec:discs}

\subsection{Discussion of definition \ref{dfn:e-val-approx}}\label{sec:dfn-arguments}

In Section \ref{sec:e-var-approx} we introduced the notion of $e$-variable approximability.
This notion might not seem very intuitive, so let us consider a few examples when seemingly more intuitive ways to combine $e$-variables fail.

First, we will demonstrate that using a simple $\inf$-projection ($x \mapsto \inf_{\theta}e_\theta(x)$) might lead to non-measurable functions. 

\begin{example}
    Let us consider a family of normal distributions $\{\mathcal{N}(\mu,1)\}_{\mu\in\reals}$ and some non-measurable set $L \subset \reals$. For each $\mu\in L$ let $e_\mu(x) = \mathbb{I}\{x\ne \mu\}$, and for all $\mu \notin L$ let $e_{\mu} \equiv 1$. Then all the functions $e_\mu$ are measurable with $\expectE_\mu e_\mu(X)=1$, but the function $\inf_\mu e_\mu(x) = \mathbb{I}\{x\notin L\}$ is not measurable and thus cannot serve as a statistical test.
\end{example}

\begin{remark}
    Taking an infimum over a countable set would preserve measurability, but it is weaker (up to a multiplicative constant) than our proposed approach and, what is also important for us, does not preserve lower semi-computability. 
\end{remark}

The next two example shows that for non-trivial families we cannot use maximum likelihood estimator for $\hat{s}$

\begin{example}
    Let us consider a family of normal distributions $\{\mathcal{N}(\mu,1)\}_{\mu\in\reals}$ and an arbitrary constant $C>)$. For each value of $\mu$ we define a function $e_\mu(x) = C\cdot10^{100} \mathbb{I}\{x=\mu\}$. Then, obviously, each $e_\mu$ is measurable and $\expectE_{\mu}e_{\mu}(X) =0\le1$, hence all $e_\mu$ are $e$-variables, but at the same time for all $\mu$ we have $\expectE_\mu \frac{1}{C}e_{\operatorname{MLE}(x)}(x)=\frac{1}{C}\expectE_\mu e_x(x) = 10^{100} > 1$.
\end{example}

\begin{example}
    Let us consider a family of Poisson distributions $\{\Pois(\lambda)\}_{\lambda>0}$, then $\operatorname{MLE}(x)=x$. For each natural number $n$ define a function $e_n(x) = \frac{1}{\Pois(n)(\{n\})} \mathbb{I}\{x=n\}=\exp(n)\frac{n!}{n^n}\mathbb{I}\{x=n\}$. Then each $e_n$ is measurable and $\expectE_{n}e_{n}(X) = 1$, hence all $e_n$ are $e$-variables, but at the same time for all $\lambda>0$ we have $\expectE_\lambda e_{\operatorname{MLE}(X)}(X)=\expectE_\lambda e_X(X) = \exp(X) \frac{X!}{X^X} > 1$.

    Moreover, since $\exp(x) \frac{x!}{x^x} \sim \sqrt{2\pi x}$, for any $C>0$ there will be a large enough $\lambda$ for which $\expectE_\lambda \frac{1}{C}e_X(X) > 1$.
\end{example}

\subsection{Further research}

There are several directions which we think are worth investigating further.

In theorem \ref{thm:most-e-var-results} we see that there may be many different nets for the same family of distributions and even many different estimators for the same net. This naturally leads to the question of characterising the nets and finding which nets and estimators allow the optimal normalising factor.

Estimators in our work are based on the idea that the data represents information about the distribution that is distorted by random noise, so we need to ``trim'' the data in order to reduce the impact of the noise while leaving useful information. This view is closely related to the Minimum Description Length principle, and it might be interesting to explore further how theory behind MDL principle can be used to improve our results.

\printbibliography

@article{RamdasWang2025,
url = {http://dx.doi.org/10.1561/3600000002},
year = {2025},
volume = {1},
journal = {Foundations and Trends® in Statistics},
title = {Hypothesis Testing with E-values},
doi = {10.1561/3600000002},
issn = {2978-4212},
number = {1-2},
pages = {1-390},
author = {Aaditya Ramdas and Ruodu Wang}
}

@misc{gacs2021lecturenotes,
      title={Lecture notes on descriptional complexity and randomness}, 
      author={Peter Gacs},
      year={2021},
      eprint={2105.04704},
      archivePrefix={arXiv},
      primaryClass={cs.IT},
      url={https://arxiv.org/abs/2105.04704}, 
}

@inproceedings{levin1976uniform,
  title={Uniform tests of randomness},
  author={Levin, Leonid Anatolevich},
  booktitle={Doklady Akademii Nauk},
  volume={227},
  number={1},
  pages={33--35},
  year={1976},
  organization={Russian Academy of Sciences}
}

@article{LEVIN198415,
title = {Randomness conservation inequalities; information and independence in mathematical theories},
journal = {Information and Control},
volume = {61},
number = {1},
pages = {15-37},
year = {1984},
issn = {0019-9958},
doi = {https://doi.org/10.1016/S0019-9958(84)80060-1},
url = {https://www.sciencedirect.com/science/article/pii/S0019995884800601},
author = {Leonid A. Levin}
}

@article{GACS200591,
title = {Uniform test of algorithmic randomness over a general space},
journal = {Theoretical Computer Science},
volume = {341},
number = {1},
pages = {91-137},
year = {2005},
issn = {0304-3975},
doi = {https://doi.org/10.1016/j.tcs.2005.03.054},
url = {https://www.sciencedirect.com/science/article/pii/S030439750500188X},
author = {Peter Gács}
}

@article{VovkVyuginBayes,
    author = {Vovk, V. G. and V'Yugin, V. V.},
    title = "{On the Empirical Validity of the Bayesian Method}",
    journal = {Journal of the Royal Statistical Society: Series B (Methodological)},
    volume = {55},
    number = {1},
    pages = {253-266},
    year = {2018},
    month = {12},
    issn = {0035-9246},
    doi = {10.1111/j.2517-6161.1993.tb01482.x},
    url = {https://doi.org/10.1111/j.2517-6161.1993.tb01482.x},
    eprint = {https://academic.oup.com/jrsssb/article-pdf/55/1/253/49173327/jrsssb\_55\_1\_253.pdf},
}

@incollection{shen,
  TITLE = {{Randomness Tests: Theory and Practice}},
  AUTHOR = {Shen, Alexander},
  URL = {https://hal-lirmm.ccsd.cnrs.fr/lirmm-03065320},
  BOOKTITLE = {{Fields of Logic and Computation III}},
  EDITOR = {Blass A. and C{\'e}gielski P. and Dershowitz N. and Droste M. and Finkbeiner B.},
  PUBLISHER = {{Springer-Verlag}},
  SERIES = {Lecture Notes in Computer Science},
  VOLUME = {12180},
  PAGES = {258-290},
  YEAR = {2020},
  MONTH = May,
  DOI = {10.1007/978-3-030-48006-6\_18},
  HAL_ID = {lirmm-03065320},
  HAL_VERSION = {v1},
urldate={2024-05-28},
}

@article{Kolmogorov1968LogicalBF,
  title={Logical basis for information theory and probability theory},
  author={Andrei N. Kolmogorov},
  journal={IEEE Trans. Inf. Theory},
  year={1968},
  volume={14},
  pages={662-664},
  url={https://api.semanticscholar.org/CorpusID:11402549},

urldate={2024-05-28},
}

@misc{vovk2016concept,
      title={On the concept of Bernoulliness}, 
      author={Vladimir Vovk},
      year={2016},
      eprint={1612.08859},
      archivePrefix={arXiv},
      primaryClass={math.ST},
 urldate = {2024-05-28},
}

@incollection{vovk2020nonalgorithmic,
  title={Non-algorithmic theory of randomness},
  author={Vovk, Vladimir},
  booktitle={Fields of Logic and Computation III: Essays Dedicated to Yuri Gurevich on the Occasion of His 80th Birthday},
  pages={323--340},
  year={2020},
  publisher={Springer}
}

@article{Bienvenu2011AlgorithmicTA,
  title={Algorithmic tests and randomness with respect to a class of measures},
  author={Laurent Bienvenu and P{\'e}ter G{\'a}cs and Mathieu Hoyrup and Cristobal Rojas and Alexander Shen},
  journal={Proceedings of the Steklov Institute of Mathematics},
  year={2011},
  volume={274},
  pages={34-89},
  url={https://api.semanticscholar.org/CorpusID:3178723}
}

@article{VovkLearningBernoulli199796,
title = {Learning about the Parameter of the Bernoulli Model},
journal = {Journal of Computer and System Sciences},
volume = {55},
number = {1},
pages = {96-104},
year = {1997},
issn = {0022-0000},
doi = {https://doi.org/10.1006/jcss.1997.1502},
url = {https://www.sciencedirect.com/science/article/pii/S0022000097915026},
author = {V.G Vovk}
}

@article{VovkWang2021Calibration,
author = {Vladimir Vovk and Ruodu Wang},
title = {{E-values: Calibration, combination and applications}},
volume = {49},
journal = {The Annals of Statistics},
number = {3},
publisher = {Institute of Mathematical Statistics},
pages = {1736 -- 1754},
year = {2021},
doi = {10.1214/20-AOS2020},
URL = {https://doi.org/10.1214/20-AOS2020}
}

@book{weihrauch2000computable,
  title={Computable Analysis: An Introduction},
  author={Weihrauch, K.},
  isbn={9783540668176},
  lccn={lc00056310},
  series={Texts in Theoretical Computer Science. An EATCS Series},
  url={https://books.google.co.uk/books?id=OPolVWVFDJYC},
  year={2000},
  publisher={Springer Berlin Heidelberg}
}

@article{GrunwaldSafeTesting,
    author = {Grünwald, Peter and de Heide, Rianne and Koolen, Wouter},
    title = {Safe testing},
    journal = {Journal of the Royal Statistical Society Series B: Statistical Methodology},
    volume = {86},
    number = {5},
    pages = {1091-1128},
    year = {2024},
    month = {03},
    issn = {1369-7412},
    doi = {10.1093/jrsssb/qkae011},
    url = {https://doi.org/10.1093/jrsssb/qkae011},
    eprint = {https://academic.oup.com/jrsssb/article-pdf/86/5/1091/60648648/qkae011.pdf},
}

@ARTICLE{LardyReverseInfProj,
  author={Lardy, Tyron and Grünwald, Peter and Harremoës, Peter},
  journal={IEEE Transactions on Information Theory}, 
  title={Reverse Information Projections and Optimal E-Statistics}, 
  year={2024},
  volume={70},
  number={11},
  pages={7616-7631},
  doi={10.1109/TIT.2024.3444458}}

\appendix

\section{Notaion}

\begin{itemize}
    \item $\e$ stands for the constant $\e=2.71\ldots$ to distinguish from $e$, which we will use in namings of $e$-variables
    \item $|A|$ and $\#A$ are interchangeable and stand for the cardinality of a set $A$
    \item $|x|$ stands for the absolute value of a number $x$
    \item $\|x\|$ stands for the Euclidean norm of a point $x \in \reals^n$
    \item $x_{1:n}$ stands for a point in $\reals^n$ with coordinates $x_1, x_2, \ldots, x_n$.
    \item $a \le^+ b$ stands for $a \le b + C$ for some numerical constant $C$, similarly for $\ge^+$
    \item $a \le^\times b$ stands for $a \le Cb$ for some numerical constant $C>0$, similarly for $\ge^\times$
    \item $a =^+ b$ means that both $a \le^+ b$ and $b \le^+ a$ hold, similarly for $a =^\times b$
    \item $[l:r]$ stands for the set $\{n \in\intnums \mid l \le n \le r\}$
    \item $A^\varepsilon$ is the $\varepsilon$-neighbourhood of a set $A$
    \item $\mathbb{I}_A$ stands for the indicator function of a set $A$
    \item $\mathbb{I}_A^\varepsilon$ stands for the bounded continuous approximation of $\mathbb{I}_A$ which is $1$ on $A$, $0$ on the complement of $A^\varepsilon$ and linearly interpolates between $1$ and $0$ on $A^\varepsilon\setminus A$
    \item $\operatorname{round}_{\mathcal{S}}$ stands for the function that maps a real number to the closest element of a discrete set $\mathcal{S}\subset\reals$; when there are two such elements, $\operatorname{round}_\mathcal{S}$ could be either of those for the purposes of our paper
    \item $\kl(\theta_1\|\theta_2)$ stands for the Kullback–Leibler divergence $\KL(\probP_{\theta_1}\|\probP_{\theta_2})$ when the family $\{\probP_\theta\}_{\theta\in\Theta}$ is uniquely determined by the context
    \item $\expectE_\theta$ stands for $\expectE_{\probP_\theta}$ when the family $\{\probP_\theta\}_{\theta\in\Theta}$ is uniquely determined by the context
    \item $\mathfrak{P}(\mathcal{X})$ stands for the set of all probability measures on $\mathcal{X}$
    \item $\uut$ denotes a universal uniform randomness test, arbitrarily chosen and fixed for this paper
    \item the index set is always specified in an indexed set notation to prevent confusion with singleton sets, that is, notation $\{\probP_\theta\}$ always stands for a singleton set and is never an abbreviation of $\{\probP_\theta\}_{\theta\in\Theta}$
\end{itemize}

\section{Proofs about $e$-variable-approximability}

\subsection{Proofs of lemmas}

\begin{proof}[Proof of lemma \ref{lem:4-cond-techn-lemma}]
    We want to show that $\frac{1}{C}e_{\hat s (X)}(X)$ is an $e$-variable. That is, we need to show that, for all $\theta \in \Theta$
    \[
        \int_{\reals^n} e_{\hat s(x)}(x) \probP_\theta({\rm d}x)\leq C.
    \]
    First, for each $s \in \mathcal{S}$ we denote $\mathcal{X}_s = \hat s^{-1}(s)\subseteq \reals^n$. Now,
    \begin{align*}
        \int_{\reals^n} e_{\hat s(x)}(x) \probP_\theta({\rm d}x) 
        &=
        \sum_{s\in\mathcal{S}} \int_{\mathcal{X}_s}e_s(x) p_{\theta}(x) {\rm d}x
        \\ &=
        \sum_{s\in\mathcal{S}} \int_{\mathcal{X}_s}
            \frac{p_\theta(x)}{p_s(x)}
            e_s(x)p_s(x){\rm d}x
        \\ &\overset{\text{by \ref{prop:loglike-to-d}}}{=}
        \sum_{s\in\mathcal{S}} \int_{\mathcal{X}_s}
            \exp\Big[d(g(x)\|s) - d(g(x)\|\theta)\Big]
            e_s(x)p_s(x){\rm d}x 
        \intertext{denote 
            $C_{s}=C_{s,\theta}   =\sup_{x\in\mathcal{X}_s}d(g(x)\|s) - d(g(x)\|\theta)$}
        &\le\sum_{s\in\mathcal{S}}
            \exp(C_s)\int_{\mathcal{X}_s}
            e_s(x)p_s(x){\rm d}x 
        \\ &\le\sum_{s\in\mathcal{S}}
            \exp(C_s)\underbrace{\int_{\reals^n}e_s(x)p_s(x){\rm d}x}_{\le 1,\textrm{ by def. of $e$-var.}}
        \\ &\le\sum_{s\in\mathcal{S}}
            \exp(C_s).
    \end{align*}
    We now have to estimate the sum $\sum_{s\in\mathcal{S}}
    \exp(C_s).$

    By \ref{prop:near-bound} we have 
    \[
        C_{s,\theta} \leq c' - \inf_{x\in\mathcal{X}_s} d(g(x)\|\theta),
    \]
    and, from non-negativity of $d$, it follows that
    $
        C_s \leq c'.
    $

    Whenever $s \ge \operatorname{succ}_{\mathcal{S}}^{(k+1)}(\theta)$ for some $k$, then $\operatorname{pred}_{\mathcal{S}}(s) \ge \operatorname{succ}_{\mathcal{S}}^{(k)}(\theta)$, and, by \ref{prop:g-near-s} for all $x \in \mathcal{X}_s$:
    \[
        g(x) > \operatorname{pred}_{\mathcal{S}}(g(x)) \ge \operatorname{pred}_{\mathcal{S}}(s) \ge \operatorname{succ}_{\mathcal{S}}^{(k)}(\theta) > \theta,
    \]
    and $\left|\mathcal{S} \cap (g(x), \theta) \right| \ge k$, and, by \ref{prop:d-grows-fast-enough} we have: 
    \[
        d(g(x)\| \theta) \ge (1+\alpha)\log (k-1), ~~\forall x\in\mathcal{X}_s,
    \]
    and so, in this case $C_{s, \theta} \le c' - (1+\alpha)\log (k-1)$.

    Similarly, we get the same bound for $s \le \operatorname{pred}_{\mathcal{S}}^{(k+1)}(\theta)$.
    
    Putting this all together, we get
    \begin{align*}
        \sum_{s\in\mathcal{S}}
            \exp(C_{s,\theta}) &\le
            \left(\sum_{\substack{s \le \operatorname{pred}_{\mathcal{S}}^{(3)}(\theta) \\ s\in\mathcal{S}}}
            +
            \sum_{\substack{\operatorname{succ}_{\mathcal{S}}^{(3)}(\theta) \le s \\ s\in\mathcal{S}}}
            +
            \sum_{\substack{s\in \mathcal{S} \\ \operatorname{pred}_{\mathcal{S}}^{(3)}(\theta) < s \\ s < \operatorname{succ}_{\mathcal{S}}^{(3)}(\theta)}}
            \right) \exp(C_{s,\theta})
            \\
            &\le \sum_{k=1}^\infty \e^{c'-(1+\alpha)\log k} + \sum_{k=1}^\infty e^{c'-(1+\alpha)\log k} + 5\e^{c'} 
            \\
            &= \e^{c'} \left(5 + 2\sum_{k=1}^{\infty} \frac{1}{k^{1+\alpha}}\right) 
            \\
            &\le
            e^{c'} \left(5 + 2\left(1 + \frac{1}{\alpha}\right)\right)
            \\
            &=
            \exp(c')\left(7 + \frac{2}{\alpha}\right),
    \end{align*}
    which concludes the proof.
\end{proof}

\begin{proof}[Proof of lemma \ref{lem:5-cond-techn-lemma}]
    The proof follows the proof of lemma \ref{lem:4-cond-techn-lemma} exactly until after the establishing that
    $
        C_{s,\theta} \leq c' - \inf_{x\in\mathcal{X}_s} d(g(x)\|\theta),
    $
    and
    $
        C_s \leq c'.
    $
    
    From \ref{prop:obtuse-pyth} and non-negativity of $d$ we also have that, for any $\theta_1' \le \theta_1 \le \theta_2 \le \theta_2'$:
    \[
        d(\theta_1\|\theta_2) \le d(\theta_1' \| \theta_2'),~~
        d(\theta_2\|\theta_1)\le d(\theta_2'\|\theta_1').
    \]

    For $\operatorname{succ}_{\mathcal{S}}(s)<\operatorname{pred}_{\mathcal{S}}(\theta)\le\theta$, we have, from \ref{prop:g-near-s} and repeated application of \ref{prop:obtuse-pyth}:
    \begin{align*}
        d(g(x)\|\theta) &\ge d(\operatorname{succ}_{\mathcal{S}}(g(x))\|\operatorname{pred}_{\mathcal{S}}(\theta))
        \\
        &\ge\sum_{\substack{
                t\in \mathcal{S}\\
                \operatorname{succ}_{
                \mathcal{S}}(g(x))\le t\\ t < \operatorname{pred}_{\mathcal{S}}(\theta)
                }} 
        d(t\|\operatorname{succ}_{\mathcal{S}}(t))
        \\
        &\overset{\ref{prop:fixed-steps}}{\ge} \sum_{\substack{
                t\in \mathcal{S}\\
                \operatorname{succ}_{
                \mathcal{S}}(g(x))\le t\\ t < \operatorname{pred}_{\mathcal{S}}(\theta)
                }} c
        \\
        &\overset{\ref{prop:g-near-s}}{\ge} \sum_{\substack{
                t\in \mathcal{S}\\
                \operatorname{succ}_{
                \mathcal{S}}(s)\le t\\ t < \operatorname{pred}_{\mathcal{S}}(\theta)
                }} c
        \\
        &\ge \left(|\mathcal{S}\cap (s,\theta)|-1\right)c.
    \end{align*}

    Similarly, for $\theta \le \operatorname{succ}_{\mathcal{S}}(\theta)<\operatorname{pred}_{\mathcal{S}}(s)$, we have
    \[
        d(g(x)\|\theta) \ge (|\mathcal{S}\cap (\theta, s)|-1)c.
    \]

    Summarising, we get
    \begin{align*}
        \sum_{s\in\mathcal{S}}
            \exp(C_{s,\theta}) &\le
            \left(\sum_{\substack{s \le \operatorname{pred}_{\mathcal{S}}^{(2)}(\theta) \\ s\in\mathcal{S}}}
            +
            \sum_{\substack{\operatorname{succ}_{\mathcal{S}}^{(2)}(\theta) \le s \\ s\in\mathcal{S}}}
            +
            \sum_{\substack{s\in \mathcal{S} \\ \operatorname{pred}_{\mathcal{S}}^{(2)}(\theta) < s \\ s < \operatorname{succ}_{\mathcal{S}}^{(2)}(\theta)}}
            \right) \exp(C_{s,\theta})
            \\
            &\le \sum_{k=0}^\infty \e^{c'-kc}  + \sum_{k=0}^\infty \e^{c'-kc} + 3\e^{c'}
            \\
            &\le \frac{2\e^{c'+c}}{\e^c - 1} + 3\e^{c'},
            \\
            &=
            \exp(c')\left(
            3 + 2\left(
                    1+\frac{1}{\e^c - 1}
                \right)
            \right)
            \\
            &= \exp(c') \left(5 + \frac{2}{\e^c - 1} \right)
            ,
    \end{align*}
    which concludes the proof.
\end{proof}

\begin{proof}[Proof of lemma \ref{lem:exp-fams-are-good}]
    We will use the fact that for a canonical exponential family 
    \begin{align*}
        \kl(\eta \| \eta') =& (\eta - \eta')\expectE_{\eta}T(X) - (A(\eta) - A(\eta'))\\
        =& (\eta - \eta')\frac{{\rm d}}{{\rm d}\eta}A(\eta) - (A(\eta) - A(\eta')).
    \end{align*}

    1. Notice that
    \begin{align*}
        \log \frac{p_\eta(x)}{p_{\eta'}(x)} 
        =& \log \frac{h(x)\exp\left(\eta T(x) - A(\eta)\right)}{h(x)\exp\left(\eta' T(x) - A(\eta')\right)}
        \\
        =& (\eta - \eta')T(x) - (A(\eta) - A(\eta'))\\
        =& (\eta - \eta')\expectE_{\hat\eta (x)}T(X) - (A(\eta) - A(\eta'))\\
        =& (\hat\eta(x) - \eta')\expectE_{\hat\eta (x)} T(X) - (A(\hat\eta(x)) - A(\eta'))\\
        &- \big[ (\hat\eta(x) - \eta)\expectE_{\hat\eta (x)} T(X) - (A(\hat\eta(x)) - A(\eta)) \big]\\
        =& \kl(\hat\eta(x) \| \eta') - \kl(\hat\eta(x) \| \eta)
    \end{align*}

    2. We employ the pythagorean equality for Bregman divergence $\mathcal{B}_{A}(\eta' \| \eta)=\kl(\eta \| \eta')$:
    \begin{align*}
        \kl(\eta_3 \| \eta_1) - \kl(\eta_2 \| \eta_1) - \kl(\eta_3 \| \eta_2) &=
        \mathcal{B}_{A}(\eta_1\|\eta_3) - \mathcal{B}_{A}(\eta_1 \| \eta_2) - \mathcal{B}_{A} (\eta_2 \| \eta_3)
        \\&=(\eta_1 - \eta_2)\left(\left(\frac{\rm d}{{\rm d}\eta}A\right)(\eta_2) - \left(\frac{\rm d}{{\rm d}\eta}A\right)(\eta_3)\right)
    \end{align*}
    Since the cumulant function $A$ is convex, the function $\frac{\rm d}{{\rm d}\eta}A$ is non-decreasing, so, whenever $\eta_1 \le \eta_2 \le \eta_3$ or $\eta_1 \ge \eta_2 \ge \eta_3$, the differences $(\eta_1 - \eta_2)$ and $\left(\left(\frac{\rm d}{{\rm d}\eta}A\right)(\eta_2) - \left(\frac{\rm d}{{\rm d}\eta}A\right)(\eta_3)\right)$ are of the same sign, meaning
    \[
        \kl(\eta_3 \| \eta_1) \ge \kl(\eta_3 \| \eta_2) +\kl(\eta_2 \| \eta_1),
    \]
    which completes the proof.
\end{proof}

\subsection{Proof of Theorem \ref{thm:most-e-var-results}}

\begin{proof}[Proof for uniform distributions]

For both discrete and continuous uniform  distribution families the proof can be done directly. 

For $2^n < x \le 2^{n+1}$ we get $\hat{s}(x)=2^{n+1}$, and

\begin{align*}
    \expectE_N e_{\hat{s}(X)}(X) &= \sum_{x=0}^N \frac{1}{N+1} e_{\hat{s}(x)}(x)\\
    &= \sum_n \sum_{2^n < x \le 2^{n+1}} \frac{1}{N+1} e_{\hat{s}(x)}(x)\\
    &\le \sum_n \frac{2^{n+1}+1}{N+1} \underbrace{\sum_{x\le2^{n+1}}\frac{1}{2^{n+1}+1}e_{2^{n+1}}(x)}_{\expectE_{2^{n+1}}e_{2^{n+1}}(X)\le1}\\
    &\le \sum_n \frac{2^{n+1}+1}{N+1} \le 3,
\end{align*}

Meaning that $\expectE_N \frac{1}{3}e_{\hat{s}(X)}(X) \le 1$ for all $N$.

The proof for continuous case is essentially the same.
\end{proof}

\begin{proof}[Proof for Poisson distributions]
	The density function $p_{\lambda}(n)$ of $\Pois(\lambda)$ is
	\[
		p_\lambda(n) = \e^{-\lambda} \frac{\lambda^n}{n!} = \frac{1}{n!}\exp(n \log \lambda - \lambda).
	\]
	To apply Lemma \ref{lem:exp-fams-are-good}, we are switching to parametrisation $\eta = \log \lambda$, that is, $\expectE_\eta f = \expectE_{\Pois(\e^\eta)} f$.
	
	Notice that $n = \expectE_{\log n} X$, and so by Lemma \ref{lem:exp-fams-are-good} we have \ref{prop:loglike-to-d} and \ref{prop:obtuse-pyth} for $d = \kl$ and $g = \log$.
	
	When parametrising distributions with $\eta$, we have need to work with $\mathcal{S}_\Eta = \{2 \log t\}_{t \in \natnums_{\ge 1}}$.  We also have
	\[
		\kl(\eta_1 \| \eta_2) =
		\e^\eta_1(\eta_1 - \eta_2) - (\e^{\eta_1} - \e^{\eta_2}) = \lambda_1 \log\frac{\lambda_1}{\lambda_2} - (\lambda_1-\lambda_2) = \kl_\lambda(\lambda_1 \| \lambda_2).
	\]
	
	To prove \ref{prop:near-bound}, take any $t^2$ and $n \in \hat{s}^{-1}(t^2)$. Notice that $| n - t^2| \le t$. We have:
	\begin{align*}
		\kl( \log n \| \log t^2) &= \kl_\lambda (n \| t^2)
		\\
		&= n \log \frac{n}{t^2} - (n - t^2)
		\\
		&= n \log \left(1 + \frac{n - t^2}{t^2}\right) - (n - t^2)
		\\
		&\le \left(\frac{n}{t^2} - 1\right) (n - t^2)
		\\
		&= \frac{(n - t^2)^2}{t^2}
		\\
		&\le 1.
	\end{align*}
	
	Property \ref{prop:g-near-s} holds trivially since $\hat{s}$ sends a number to one of the nearest complete squares --- elements of $\mathcal{S}$.
	
	Property \ref{prop:fixed-steps} is proven similarly to \ref{prop:near-bound}.
	
	\begin{align*}
		\kl(\log t^2 \| \log (t+1)^2) &= t^2 \log \frac{t^2}{(t+1)^2} - (t^2 - (t+1)^2)
		\\
		&=-2t^2 \log \left(1 + \frac{1}{t}\right) + 2t + 1
		\\
		&\ge -2t^2 \frac{1}{t} + 2t + 1
		\\&=1.
	\end{align*}
	\begin{align*}
		\kl(\log (t+1)^2 \| \log t^2) &= (t+1)^2 \log \frac{(t+1)^2}{t^2} - ((t+1)^2 - t^2)
		\\
		&= 2(t+1)^2 \log\left(1 + \frac{1}{t}\right) - 2t - 1
		\\
		&\ge 2(t+1)^2 \left(\frac{1}{t} - \frac{1}{2t^2}\right) - 2t - 1
		\\
		&= \frac{(t+1)^2(2t - 1) - (2t+1)t^2}{t^2}
		\\
		&= \frac{2t^2 - 1}{t^2} \ge 1.
	\end{align*}
	
	Since we now have all \ref{prop:loglike-to-d}--\ref{prop:g-near-s}, \ref{prop:obtuse-pyth}, \ref{prop:fixed-steps}, $e$-value approximability of $\{\Pois(\lambda)\}_{\lambda > 0}$ follows from Lemma \ref{lem:5-cond-techn-lemma}.
\end{proof}

\begin{proof}[Proof for normal distributions with fixed variance.]
    Density $p_\mu$ of $\mathcal{N}(\mu,1)^{\otimes n}$ is 
    \[
        p_\mu(x_{1:n}) = \frac{1}{\sqrt{2\pi}}\e^{-\|x_{1:n}^2\|/2} \exp\left( \mu \sum_k x_k - \frac{n\mu^2}{2}\right),
    \]
    and, since $\sum_k x_k = \expectE_{\frac{\sum x_k}{n}}\sum_k X_k$ from Lemma \ref{lem:exp-fams-are-good} we get properties \ref{prop:loglike-to-d} and \ref{prop:obtuse-pyth} for $d(\theta_1\|\theta_2)=\kl(\theta_1\|\theta_2)=\frac{n}{2}(\theta_1-\theta_2)^2$ and $g(x_{1:n}) = \frac{\sum_k x_k}{n}$.

    To prove property \ref{prop:near-bound}, we note that $x_{1:n} \in \hat{s}^{-1}(s)$ means that $\left|\frac{\sum_k x_k}{n} - s\right| \le \frac{\alpha}{2\sqrt{n}}$, and so
    \[
        d(g(x_{1:n})\|s) = \frac{n}{2}(g(x_{1:n})-s)^2 \le \frac{n}{2} \frac{\alpha^2}{4n} = \frac{\alpha^2}{8}.
    \]
    
    Since $\hat{s}(x)=\operatorname{round}_{\mathcal{S}}(g(x))$, the property \ref{prop:g-near-s} is trivial.

    The difference between the consecutive elements $s, s'$ of the net is exactly $\frac{\alpha}{\sqrt{n}}$, so
    \[
        d(s\|s') = \frac{n}{2} \frac{\alpha^2}{n} = \frac{\alpha^2}{2},
    \]
    which establishes property \ref{prop:fixed-steps} and finishes the proof.
\end{proof}

\begin{proof}[Normal distribution with zero mean value]
	We have
	\[
		p_{\sigma^2}(x_{1:n})=\frac{1}{\sqrt{2\pi}}\exp\left(-\frac{\|x_{1:n}\|^2}{2\sigma^2}-\frac{n}{2}\log \sigma^2\right).
	\]
	Since $\|x_{1:n}\|^2 = \sum_k x_k^2$ and all the coordinates are i.i.d., we have $T(x_{1:n})=\expectE_{g(x_{1:n})}T(X_{1:n})$ for $T(x_{1:n}) = -\frac{1}{2}\sum_k x_k^2$ and $g(x_{1:n})= \frac{\|x_{1:n}\|^2}{n}$. 
	                                                                                                                                                                                                                   
	As canonical parametrisation $\eta = 1/\sigma^2$ reverses order, we still get the reverse triangle inequality \ref{prop:obtuse-pyth} when defining $d = \kl$ in terms of variance, and thus can continue to work in standard parametrisation.
	
	For two variances $\sigma^2_1, \sigma^2_2$, we have
	\[
		\kl(\sigma^2_1 \| \sigma^2_2) = \frac{n}{2} \left( \frac{\sigma^2_1}{\sigma^2_2}-\log \frac{\sigma^2_1}{\sigma^2_2} -1 \right).
	\]
	
	Since $\hat{s}(x) = \operatorname{round}_\mathcal{S}(g(x))$, we immediately get \ref{prop:g-near-s}. We also have that for any $s=\hat{s}(x)$ that $g(x)/s \in \left[\frac{\sqrt{n}}{1+\sqrt{n}}, \frac{1+\sqrt{n}}{\sqrt{n}}\right]$,  and for consecutive $s, s' \in \mathcal{S}$ we have $s/s' \in \left\{\frac{\sqrt{n}}{1+\sqrt{n}}, \frac{1+\sqrt{n}}{\sqrt{n}} \right\}$.
	
	Let $h(t)  = \frac{n}{2} \left(t - \log t - 1\right)$, then $\kl(\sigma^2_1 \| \sigma^2_2) = h\left(\sigma^2_1 / \sigma^2_2\right)$.
	
	First, we will show \ref{prop:fixed-steps}. 
	\begin{align*}
		h\left(\frac{1 + \sqrt{n}}{\sqrt{n}}\right) &= \frac{n}{2}\left(1 + \frac{1}{\sqrt{n}} - \underbrace{\log \left(1 + \frac{1}{\sqrt{n}} \right)}_{\le \frac{1}{\sqrt{n}} - \frac{1}{4n}} - 1\right)\\
		&\ge \frac{n}{2} \left(\frac{1}{\sqrt{n}} - \frac{1}{\sqrt{n}} + \frac{1}{4n} \right)\\
		&\ge \frac{1}{8},\\
		h\left(\frac{\sqrt{n}}{1+\sqrt{n}}\right) &= \frac{n}{2}\left(1 - \frac{1}{1+\sqrt{n}} - \underbrace{\log \left(1 - \frac{1}{1+\sqrt{n}} \right)}_{\le -\frac{1}{1+\sqrt{n}} - \frac{1}{4\left(1+\sqrt{n}\right)^2}} - 1\right)\\
		&\ge \frac{n}{2} \left(-\frac{1}{1+\sqrt{n}} + \frac{1}{1 + \sqrt{n}} + \frac{1}{4\left(1 + \sqrt{n}\right)^2}\right)\\
		&\ge \frac{1}{32},
	\end{align*}
	so for any two consecutive elements $s, s'$ of the net we have $d(s \| s') \ge 1/32$.
	
	Now, to show \ref{prop:near-bound}, let $t$ be between $ - \frac{1}{1+\sqrt{n}}$ and $\frac{1}{\sqrt{n}}$.
	\begin{align*}
		h(1+t) &= \frac{n}{2}((1+t) - \underbrace{\log (1+t)}_{\ge t - t^2} - 1)
		\\&\le \frac{n}{2}(t - t + t^2)\\
		&\le \frac{|t|^2 n}{2}
		\\&\le \frac{1}{n}\frac{n}{2} = \frac{1}{2}
	\end{align*}
\end{proof}

\begin{proof}[Proof for Cauchy distributions.]
    For Cauchy distributions, we have $p_\theta(x)=\frac{1}{\pi(1+(x-\theta)^2)}$, and so
    \[
        \log \frac{p_\theta(x)}{p_s(x)} = \log \left(1+(x-s)^2\right) - \log \left(1+(x-\theta)^2\right).
    \]
    Take $g(x)=x$, $d(\theta_1 \| \theta_2) = \log\left(1 + (\theta_1 - \theta_2)^2 \right) \ge \log1=0$.

    For $\hat{s}=r^\varepsilon$, we have that $x \in \hat{s}^{-1}(n) \implies x \in [n-0.5-\varepsilon, n +0.5+\varepsilon]$, so 
    \[
        d(x\|n) \le \log (1+(0.5+\varepsilon)^2)\le\log 2.
    \]

    Property \ref{prop:g-near-s} for $r^\varepsilon$ follows from the same observation in a straightforward manner.

    Now, assume there are $k$ integers between $\theta_1$ and $\theta_2$. Then $|\theta_1 - \theta_2| \ge k-1$ and
    \[ 
    d(\theta_1\|\theta_2) = \log\left(1 + (\theta_1 - \theta_2)^2 \right) \ge 2\log(k-1).
    \]

    The result now follows from Lemma \ref{lem:4-cond-techn-lemma}.
\end{proof}

\begin{proof}[Proof for normal distributions with $r^\varepsilon$.]
    The proof is mostly identical to the proof for $\operatorname{round}_\mathcal{S}$ with a slightly weaker bounds.
\end{proof}

\section{Proofs for randomness tests}
All the lemmas in section \ref{sec:rand-def-res} follow from Theorem 4.1.1. in \cite{gacs2021lecturenotes}, which we here restate.

\begin{theorem}
	Let $\phi_e$, $e\in\natnums$, be an enumeration of all l.s.c. functions on $\mathcal{X}\times\mathfrak{M}(\mathcal{X})\times\mathcal{Y} \to [0,\infty]$. where $\mathfrak{M}(\mathcal{X})$ is the set of all bounded measures on $\mathcal{X}$. Let $b: \mathcal{Y}\times\mathfrak{M}(\mathcal{X})\to [0,\infty]$ be a l.s.c. function. There is a recursive function $e \mapsto e'$ with the following properties.
	
	For each $e$, the function $\phi_{e'}$ is everywhere defined and its integral is bounded by $b(\probP, y)$.
	
	For each $e, y, \probP$, if  $\expectE_\probP \phi_e(X\mid \probP, y) < b(\probP, y)$, then $\phi_{e'}(\cdot \mid \probP, y) =^\times \phi_e(\cdot \mid \probP, y)$.
\end{theorem}

\end{document}